\documentclass[submission,copyright,creativecommons]{eptcs}
\providecommand{\event}{QAPL 2012}

\usepackage[pdftex]{graphicx}
\usepackage[activate=true]{microtype}
\usepackage[normalem]{ulem}
\usepackage{stmaryrd}
\usepackage{amsmath}
\usepackage{amssymb}
\usepackage{amsthm}
\usepackage{wasysym}
\usepackage{setspace}
\usepackage{multirow}
\usepackage{listings}
\usepackage{url}

\newtheorem{theorem}{Theorem}

\title{Automated Verification of Quantum Protocols using {\sc mcmas}}
\author{F. Belardinelli, P. Gonzalez, A. Lomuscio
\email{\{f.belardinelli, pavel.gonzalez09, a.lomuscio\}@imperial.ac.uk}
\institute{Imperial College London\\
London, UK}}

\def\titlerunning{Automated Verification of Quantum Protocols using {\sc mcmas}}
\def\authorrunning{Belardinelli, Gonzalez, Lomuscio}

\begin{document}

\maketitle

\begin{abstract}
We present a methodology for the automated verification of quantum protocols
using {\sc mcmas}, a symbolic model checker for multi-agent systems \cite{mcmas}.
The method is based on the logical framework developed by D'Hondt and Panangaden
\cite{raqk} for investigating epistemic and temporal properties, built on the
model for Distributed Measurement-based Quantum Computation (DMC) \cite{dmbqc},
an extension of the Measurement Calculus \cite{tmc} to distributed quantum
systems. We describe the translation map from DMC to interpreted systems, the
typical formalism for reasoning about time and knowledge in multi-agent systems
\cite{fagin}. Then, we introduce {\sc dmc2ispl}, a compiler into the input
language of the {\sc mcmas} model checker \cite{mcmas}. We demonstrate the
technique by verifying the Quantum Teleportation Protocol, and discuss the
performance of the tool.
\end{abstract}

\section{Introduction}
\label{sec:introduction}

Quantum computing has gained prominence in the last decade due to theoretical
advances as well as
applications to
security, information processing, and simulation of quantum mechanical systems
\cite{nielsen}. With this increase of activity, the need for validation of
correctness of quantum algorithms has arisen. Model checking
has shown to be a promising verification technique
~\cite{cgp99}.
However, 
tools and techniques for model checking both temporal and epistemic properties of
quantum systems have not been developed yet. In this paper we aim to bridge this
gap by introducing a methodology for the automated verification of quantum
protocols using {\sc mcmas}~\cite{mcmas}, a symbolic model checker for
multi-agent systems (MAS).

The fundamental question from an epistemic point of view is how to model a flow
of quantum information. Is it meaningful to talk about ``quantum knowledge''? And
if it is, how can we express this concept? Several logics, which can be used for
reasoning about knowledge in the context of distributed quantum computation, have
been recently suggested. One of the first attempts
was based on Quantum Message Passing Environments~\cite{kiqs}. A different
approach, i.e.~Quantum Dynamic-Epistemic Logic~\cite{cafqa, lqp2, dlpoqb}, was
developed to model
the behaviour of quantum systems. A third account~\cite{ckfqcr, raqk, ckfqs} was
built on the Distributed Measurement-based Quantum Computation \cite{dmbqc},
which extends the Measurement Calculus~\cite{tmc}, a formal model for
one-way quantum computations. Among these accounts, the logic based on
Distributed Measurement-based Quantum Computation (DMC) has an underlying
operational semantics similar to the semantics of interpreted systems~\cite{
fagin}.
This makes it suitable for model checking using {\sc mcmas}.
However, interpreted systems (IS) have a Boolean semantics, which requires us to
abstract from the underlying probability distribution. While we
recognise that a full analysis of quantum phenomena requires stochastic
considerations, we believe there are still useful lessons to be learned about
protocols when these are abstracted from.
The point of the paper is partly to explore this hypotheses.

In this paper we describe a translation from DMC to IS. We
also
report on a source-to-source compiler that performs the translation into the
input language of
{\sc mcmas}. The compiler enables the use of {\sc mcmas} to
verify automatically temporal and epistemic properties of quantum protocols
specified in DMC. We verify the Quantum Teleportation protocol~\cite{qtp} against
the properties
stated and informally proved in~\cite{raqk},
and show that one specification does not hold contrary to the paper's claim.

\textbf{Related Work.} Several approaches to model checking quantum
systems have already appeared in the literature. To our knowledge,
the only dedicated
verification tool for quantum protocols is the Quantum Model Checker (QMC)
\cite{qmc}. The model checker supports specifications in quantum
computational temporal logic (QCTL), but quantum operators are
restricted to the 
Clifford group, which is the normalizer of the group of Pauli operators~\cite{nielsen}. Although
it contains many common operators, quantum circuits that involve only Clifford
group operators are not universal. Such circuits can be simulated in polynomial
time on a classical computer; however, this leads to a loss of expressive power.

In the same research line~\cite{ylyf10} a
theoretical framework to model check LTL properties using quantum
automata is proposed, and an algorithm for checking invariants of
quantum systems is presented. Finally, in~\cite{eaa10} the Quantum Key
Distribution (QKD) protocol is verified against specific eavesdropping security
properties. The authors elaborate an \emph{ad hoc} model of the
protocol, that they analyse using PRISM~\cite{prism}.

However, we stress that none of these contributions explicitly deal
with knowledge. So, these approaches do not allow the
verification of the temporal epistemic properties discussed in~\cite{raqk}.

\textbf{Structure.} Organizationally, Section~\ref{sec:preliminaries} gives an
overview of the Distributed Measurement-based Quantum Computation, Interpreted
Systems, and Quantum Epistemic Logic. Section~\ref{sec:mapping} presents a
methodology for translating a protocol specified in DMC into the corresponding
IS. Section~\ref{sec:implementation} describes and evaluates an implementation of
the formal methodology. Section~\ref{sec:conclusion} offers brief conclusions.

\section{Preliminaries}
\label{sec:preliminaries}

We discuss only the issues directly related to the paper and refer the reader to
the relevant references for an in-depth coverage of these topics. We assume
familiarity with the concepts of quantum computation~\cite{nielsen}.

\subsection{Distributed Measurement-based Quantum Computation}
\label{sub:DMC}

At the heart of the Measurement Calculus are measurement patterns~\cite{tmc}. A
\emph{pattern} $\mathcal{P} = (V,I,O,\mathcal{A})$ consists of a
\emph{computation space} $V$, which contains all qubits involved in the execution
of $\mathcal{P}$, a set $I$ of input qubits, a set $O$ of output qubits, and a
finite sequence $\mathcal{A}$ of commands $A_{p} \ldots A_{1}$, which are applied
to qubits in $V$ from right to left. The possible commands are the entanglement
operator $E_{qr}$, the measurement $M^{\alpha}_{q}$, and the corrections $X_{q}$
and $Z_{q}$, where
$q$ and $r$ represent the qubits on which these commands operate, and $\alpha$ is
a measurement angle in $[0, 2\pi]$.

An \emph{agent} \textbf{A}~\cite{dmbqc}, denoted as ${\mathbf A}({\mathbf i},
{\mathbf o}) : Q.\mathcal{E}$, is characterised by its classical input \textbf{i}
and output \textbf{o}, by a set $Q$ of qubits, and by a finite event sequence
$\mathcal{E}$, which consists of patterns and commands for classical (c?$x$,
c!$y$) and quantum (qc?$x$, qc!$q$) communication. A \emph{network} $\mathcal{N}$
of agents~\cite{dmbqc} is defined as a set of concurrently acting agents,
together with the global quantum state $\sigma$, specifically $\mathcal{N} =
{\mathbf A}_1({\mathbf i}_1, {\mathbf o}_1) : Q_1.\mathcal{E}_1 \text{ } | \text{
} \ldots \text{ } | \text{ } {\mathbf A}_m({\mathbf i}_m, {\mathbf o}_m) :
Q_m.\mathcal{E}_m \text{ } \| \text{ } \sigma$, abbreviated as $\mathcal{N} = |_i
\text{ } {\mathbf A}_i({\mathbf i}_i, {\mathbf o}_i) : Q_i.\mathcal{E}_i \text{ }
\| \text{ } \sigma$. The \emph{configuration} $C$
of a network $\mathcal{N}$ at a particular point in time is described by a set of
agents, their classical local states, and the quantum state $\sigma$, formally $C
= \sigma, \Gamma_1, {\mathbf a}_1 \text{ }|\text{ } \Gamma_2, {\mathbf a}_2
\text{ }|\text{ } \ldots \text{ }|\text{ } \Gamma_m, {\mathbf a}_m$, abbreviated
as $C = \sigma, |_i \text{ } \Gamma_i, {\mathbf a}_i$, where $\Gamma_i$
represents the classical state of agent ${\mathbf a}_i$, which is defined as a
partial mapping from classical variables to values. The set
$\mathcal{C}_\mathcal{N}$ contains all configurations that potentially occur
during the execution of the network $\mathcal{N}$.

Operational and denotational semantics for DMC are defined in~\cite{dmbqc};
however, here we are more interested in its small-step semantics. The following
small-step rules for configuration transitions describe how the network evolves
over time. If the quantum state does not change in an evaluation step, the
writing $\sigma \vdash$ precedes the rule. Also, we use a shorthand notation for
agents: ${\mathbf a}_i = {\mathbf A}_i : Q_i.\mathcal{E}_i$, ${\mathbf a}_i.E =
{\mathbf A}_i : Q_i.[\mathcal{E}_i.E]$, ${\mathbf a}^{-q} = {\mathbf A} :
Q\backslash{q}.\mathcal{E}$, and ${\mathbf a}^{+q} = {\mathbf A} : Q \uplus
{q}.\mathcal{E}[q/x]$, where $E$ is some event.

{\small
\begin{equation}
    \frac{\sigma,\mathcal{P}(V, I,O,\mathcal{A})\longrightarrow_\lambda \sigma',
        \Gamma'}{\sigma, \Gamma,{\mathbf A} : I \uplus R.[\mathcal{E}.\mathcal{P}]
        \Longrightarrow_\lambda \sigma', \Gamma \cup \Gamma',{\mathbf A} : O \uplus
        R.\mathcal{E}} \label{eq:r1}
\end{equation}
\begin{equation}
    \frac{\Gamma_2(y) = v}{\sigma \vdash (\Gamma_1, {\mathbf a}_1.\text{c?}x \text{
        }|\text{ } \Gamma_2, {\mathbf a}_2.\text{c!}y \Longrightarrow \Gamma_1[x
        \mapsto v], {\mathbf a}_1 \text{ }|\text{ } \Gamma_2, {\mathbf a}_2)}
        \label{eq:r2} \vspace{5pt}
\end{equation}
\begin{equation}
    \frac{}{\sigma \vdash (\Gamma_1, {\mathbf a}_1.\text{qc?}x \text{ }|\text{ }
        \Gamma_2, {\mathbf a}_2.\text{qc!}q \Longrightarrow \Gamma_1, {\mathbf
        a}^{+q}_1 \text{ }|\text{ } \Gamma_2, {\mathbf a}^{-q}_2)} \label{eq:r3}
\end{equation}
\begin{equation}
    \frac{L \Longrightarrow_\lambda R}{L \text{ }|\text{ }L'
        \Longrightarrow_\lambda R\text{ }|\text{ }L'} \label{eq:r4}
\end{equation}
}

The first rule refers to local operations. Since a pattern's big-step semantics
is given by a probabilistic transition system, described by $\longrightarrow$, a
probability $\lambda$ is introduced here. Also, an agent changes its sort
depending on the pattern's output $O$. The next two rules are for the classical
and the quantum rendez-vous. For the quantum rendez-vous a substitution $q$ for
$x$ in the event sequence of the receiving agent is performed and agents need to
update their qubit sorts. (4) is a metarule, which is required to express that
any of the other rules may fire in the context of a larger system.

\subsection{Interpreted Systems and {\sc mcmas}}
\label{sec:IS}

Interpreted systems 
~\cite{
fagin} are the typical formalism for reasoning about time and knowledge in
multi-agent systems.
In IS
each agent $i$ from a non-empty set $Ag$ of agents is modelled by a set of
local states $L_i$, a set of actions $Act_i$ that she may perform according to
her protocol function $P_i$, and an evolution function $t_i$.
A special agent $E$, representing the environment in which the other agents
operate, is also described by a set of local states $L_E$, a set of actions
$Act_E$, a protocol $P_E$, and an evolution function $t_E$. For every
$j \in Ag \cup \{E\}$, the protocol $P_j$ is defined as a function $P_j : L_j
\rightarrow 2^{Act_j}$, assigning a set of actions to a given local state.
Intuitively, $\alpha_j \in P_j(l_j)$ means that action $\alpha_j$ is enabled in
$l_j$. The evolution function $t_j$ is a transition function returning the target
local state given the current local state and the set of actions performed by all
agents, formally
 $t_j : L_j 
\times Act_1 \times \cdots \times Act_n \times Act_E \rightarrow L_j$ under the
constraint $\alpha_j \in P_j(l_j)$. Agents evolve simultaneously in every state
of the system according to the joint transition function $t$.

The set $Act$ of joint actions is defined as the Cartesian product of all agents'
actions, formally $Act = Act_1 \times \cdots \times Act_n \times Act_E$.
The Cartesian product $S = L_i \times \cdots \times L_n \times L_E$ of the
agents' local states
is the set of all global states of the system. The local state of agent $i$ in
the global state $g \in S$ is denoted as $l_i(g)$.
The description of an interpreted system is concluded by including a set of
atomic propositions $AP = \{p_1,p_2,\ldots\}$ and an evaluation relation $V
\subseteq AP \times S$.  Formally, an interpreted system is defined as a
tuple
 $   IS=\left\langle (L_i,Act_i,P_i,t_i)_{i\in Ag},(L_E,Act_E,P_E,t_E),V
        \right\rangle.\nonumber$

Interpreted systems can be used to 
interpret CTLK, a logic combining the branching-time temporal logic CTL with
epistemic modalities.
The formal language $\mathcal{L}$ is built from 
propositional atoms $p \in AP$ and 
agents $i \in Ag$ as follows:
{\small
\begin{equation}
    \varphi ::= p \mid \lnot \varphi \mid \varphi \lor \varphi \mid EX
    \varphi \mid EG \varphi \mid E \varphi U\psi \mid K_i 
    \nonumber
\end{equation}}
The formulae in $\mathcal{L}$ have the
following intuitive meaning. $EX\varphi$: there is a path where $\varphi$ holds
in the next state; $EG\varphi$: there is a path where $\varphi$ always holds;
$E \varphi U\psi $: there is a path where $\varphi$ holds at least until at some
state $\psi$ holds; $K_i\varphi$: agent $i$ knows $\varphi$. 
The other standard CTL formulae, e.g., $AF\varphi$: for all paths $\varphi$
eventually holds, can be derived from the above. The formal definition of
satisfaction in interpreted systems follows.

In an interpreted system $\mathcal{M}$ the evolution function $t$ determines a
transition relation $\xrightarrow[]{\mathcal{M}}$ on states such that $s
\xrightarrow[]{\mathcal{M}} s'$ iff there is a joint action $\alpha \in Act$ such
that $t(s,\alpha) = s'$.  A {\em path} $\pi$ is an infinite sequence of states
$s_0 \xrightarrow[]{\mathcal{M}} s_1 \xrightarrow[]{\mathcal{M}} \ldots$.
Further, $\pi^n$ denotes the $n$-th state in the sequence, i.e, $s_n$.  Finally,
for each agent $i \in Ag$, we introduce the epistemic equivalence relation
$\sim^{\mathcal{M}}_i$ such that $s \sim^{\mathcal{M}}_i s'$ iff $l_i(s) =
l_i(s')$.

Given the IS $\mathcal{M}$, a state $s$, and a formula $\phi \in \mathcal{L}$,
the satisfaction relation $\vDash$ is defined as follows:
\begin{tabbing}
$(\mathcal{M},s) \vDash p$ \ \ \ \ \ \ \ \ \ \ \ \= iff \ \ \= $V(p,s)$\\
$(\mathcal{M},s) \vDash \lnot \phi$ \> iff \> $(\mathcal{M},s) \not \vDash \phi$\\
$(\mathcal{M},s) \vDash \phi \lor \phi'$ \> iff \> $(\mathcal{M},s) \vDash \phi$ or $(\mathcal{M},s) \vDash \phi'$\\
$(\mathcal{M},s) \vDash E X \phi$ \> iff \> there is a path $\pi$ such that $\pi^0 = s$, and $(\mathcal{M},\pi^1) \vDash \phi$\\
$(\mathcal{M},s) \vDash E G \phi$ \> iff \> there is a path $\pi$ such that $\pi^0 = s$, and for all $n \in \mathbb{N}$, $(\mathcal{M},\pi^n) \vDash \phi$\\
$(\mathcal{M},s) \vDash E \phi U \phi'$ \> iff \> there is a path $\pi$ such that $\pi^0 = s$, for some $n \in \mathbb{N}$, $(\mathcal{M},\pi^n) \vDash \phi'$,\\
\> \> and for all $n'$, $0 \leq n' < n$ implies $(\mathcal{M},\pi^{n'}) \vDash \phi$\\
$(\mathcal{M},s) \vDash K_i \phi$ \> iff \> for all $s' \in S$, $s \sim^{\mathcal{M}}_i s'$ implies $(\mathcal{M},s') \vDash \phi$\\
\end{tabbing}

A formula $\phi \in \mathcal{L}$ is true in an IS $\mathcal{M}$, or $\mathcal{M}
\vDash \phi$, iff for all $s \in S$, $(\mathcal{M},s) \vDash \phi$.

In~\cite{mcmas} the authors present a methodology for the verification of IS
based on model checking~\cite{cgp99} via ordered binary decision diagrams. These
verification techniques have been implemented in the {\sc mcmas} model checker.
The input to the model checker is given as an Interpreted Systems Programming
Language (ISPL) program, which is essentially a machine readable IS.



\subsection{Quantum Epistemic Logic}
\label{sub:QEL}

A formal framework for reasoning about temporal and epistemic properties of
distributed quantum systems was developed in~\cite{raqk} on top of DMC.
The authors argue that quantum knowledge is not a meaningful concept, but it is
of interest to reason about classical knowledge pertaining to a quantum system. 
In this sense, the quantum information possessed by an agent concerns the qubits
she owns, the local operations she applies to these qubits, the non-local
entanglement she shares initially, and possibly prior knowledge of her local
quantum inputs. All this information is contained in her local state $\Gamma_i$
and her event sequence $\mathcal{E}_i$. Given a network $\mathcal{N}$, the
epistemic accessibility relation $\sim^{\mathcal{N}}_i$ for an agent ${\mathbf
A}_i$ is defined in~\cite{raqk} as follows: for all configurations $C = \sigma,
|_i \text{ } \Gamma_i, {\mathbf A}_i : Q_i.\mathcal{E}_i$ and $C' = \sigma', |_i
\text{ } \Gamma'_i, {\mathbf A}_i : Q'_i.\mathcal{E}'_i$ in
$\mathcal{C}_\mathcal{N}$,
$C$ and $C'$ are {\em indistinguishable to agent} ${\mathbf A}_i$, written as $C
\sim^{\mathcal{N}}_i C'$, if $\Gamma_i = \Gamma'_i$ and $\mathcal{E}_i =
\mathcal{E}'_i$. The semantics for the modal operator $K_i$ for the knowledge of
agent ${\mathbf A}_i$ is then defined in the usual way: $(C,\mathcal{N}) \vDash
K_i\varphi$ iff for all $C'$, $C' \sim^{\mathcal{N}}_i C$ implies
$(C',\mathcal{N}) \vDash \varphi$.

We now give the truth conditions for all formulae in $\mathcal{L}$ in a network
$\mathcal{N}$.
The set of atomic propositions $AP = \{x = v, x = y, \mathbf{A}_i \textrm{ has }
q, q_1 \ldots q_n = \mathinner{|{\psi}\rangle}, q_i = q_j \}$ is considered
in~\cite{ckfqcr}. In a configuration $C$ of a network $\mathcal{N}$ the truth
conditions for these atomic propositions are given as follows:
\begin{tabbing}
  $(C, \mathcal{N}) \vDash x = v$ \ \ \ \ \ \ \ \ \ \ \ \ \ \ \= iff \ \ \= there is an agent $i$ such that  $\Gamma_i(x) =  v$\\
  $(C, \mathcal{N}) \vDash x = y$ \> iff \> there are agents $i, j$ such that $\Gamma_i(x) =  \Gamma_i(y)$\\
  $(C, \mathcal{N}) \vDash \mathbf{A}_i \textrm{ has } q$ \> iff  \> $q \in Q_i$\\
  $(C, \mathcal{N}) \vDash q_1 \ldots q_n = \mathinner{|{\psi}\rangle}$ \> iff \>  $q_1 \ldots q_n = \mathinner{|{\psi}\rangle}$\\
  $(C, \mathcal{N}) \vDash q_i = q_j$ \> iff \> there is $\mathinner{|{\psi}\rangle}$ such that $\mathinner{|{\psi}\rangle} = q_i = q_j$
\end{tabbing}

In networks the small-step rules given in Section~\ref{sub:DMC} determine a
transition relation $\xrightarrow{\mathcal{N}}$ such that $C
\xrightarrow{\mathcal{N}} C'$ iff there is a rule that applied to $C$ returns
$C'$.  A {\em path} $\gamma$ is an infinite sequence of configurations
$C_0
\xrightarrow{\mathcal{N}} C_1 \xrightarrow{\mathcal{N}}
\ldots$. Further, $\gamma^n$ denotes the $n$-th state in the sequence, i.e,
$C_n$. Finally, for each agent $\mathbf{A}_i$ in the network, we introduce the
epistemic equivalence relation $\sim^{\mathcal{N}}_i$ such that $C
\sim^{\mathcal{N}}_i C'$ iff $\Gamma_i = \Gamma'_i$ and $\mathcal{E}_i =
\mathcal{E}'_i$.

Given the network $\mathcal{N}$, a configuration $C$, and a formula $\phi \in
\mathcal{L}$, the satisfaction relation $\vDash$ is defined as follows:
\begin{tabbing}
$(C, \mathcal{N}) \vDash p$ \ \ \ \ \ \ \ \ \ \ \ \= iff \ \ \= $C$ satisfies the corresponding condition above for atomic $p \in AP$\\
$(C, \mathcal{N}) \vDash \lnot \phi$ \> iff \> $(C, \mathcal{N}) \not \vDash \phi$\\
$(C, \mathcal{N}) \vDash \phi \lor \phi'$ \> iff \> $(C, \mathcal{N}) \vDash \phi$ or $(C, \mathcal{N}) \vDash \phi'$\\
$(C, \mathcal{N}) \vDash E X \phi$ \> iff \> there is a path $\gamma$ such that $\gamma^0 = C$, and $(\gamma^1, \mathcal{N}) \vDash \phi$\\
$(C, \mathcal{N}) \vDash E G \phi$ \> iff \> there is a path $\gamma$ such that $\gamma^0 = C$, and for all $n \in \mathbb{N}$, $(\gamma^n, \mathcal{N}) \vDash \phi$\\
$(C, \mathcal{N}) \vDash E \phi U \phi'$ \> iff \> there is a path $\gamma$ such that $\gamma^0 = C$, for some $n \in \mathbb{N}$, $(\gamma^n,\mathcal{N}) \vDash \phi'$,\\
\> \> and for all $n'$, $0 \leq n' < n$ implies $(\gamma^{n'},\mathcal{N}) \vDash \phi$\\
$(C, \mathcal{N}) \vDash K_i \phi$ \> iff \> for all $C' \in \mathcal{N}$, $C \sim^{\mathcal{N}}_i C'$ implies $(C', \mathcal{N}) \vDash \phi$\\
\end{tabbing}
A formula $\phi \in \mathcal{L}$ is true in a network $\mathcal{N}$, or
$\mathcal{N} \vDash \phi$, iff for all configurations $C$, $(C, \mathcal{N})
\vDash \phi$.

\subsection{Quantum Teleportation Protocol}
\label{sub:QTP}

The goal of the Quantum Teleportation Protocol (QTP) is to transmit a qubit from
one party to another with the aid of an entangled pair of qubits and classical
resources. For reasons of space we refer to~\cite{qtp}
for a detailed presentation of QTP. The DMC specification of the protocol is
given in~\cite{dmbqc} as:
{\small
\begin{equation*}
\mathcal{N}_{QTP} = {\mathbf A}
: \{1, 2\}.[(\text{c!}s_2s_1).M^{0,0}_{12}E_{12}] \text{ } | \text{ }
{\mathbf B} : \{3\}.[X^{x_2}_3Z^{x_1}_3.(\text{c?}x_2x_1)] \text{ }
\|\text{ } E_{23}.
\end{equation*}}
The informal reading is as follows: Alice ${\mathbf A}$ and Bob ${\mathbf B}$
share the entangled pair $E_{23}$ of qubits 2 and 3, and Alice wants to transmit
the input qubit $1$. In the first step, she entangles $(E_{12})$ her qubits 1 and
2. Then she measures $(M^{0,0}_{12})$ both of them. Next, she sends via classical
communication $(\text{c!}s_2s_1)$ the measurement outcomes to Bob. Upon receipt
$(\text{c?}x_2x_1)$, Bob applies corrections $(X^{x_2}_3Z^{x_1}_3)$ to his qubit
$3$ depending on these measurements. The result is that Bob's qubit $3$ is
guaranteed to be in the same state as Alice's input qubit $1$.

\section{Formal Mapping}
\label{sec:mapping}

In this section we present a methodology for translating a protocol specified in
DMC into the corresponding IS. Formally, we define a mapping $f:DMC \rightarrow
IS$, such that $f$ preserves satisfaction of formulae in the specification
language $\mathcal{L}$. First, we describe the translation of the global quantum
state and classical states of agents. Then we cover the rules in DMC.
Finally, we show that $f$ is sound.

\subsection{Classical States of Agents and Global Quantum State}
\label{sub:States}

Given a network $\mathcal{N}$ we introduce an agent $i \in Ag$ for each agent
${\mathbf A}_i({\mathbf i}_i, {\mathbf o}_i) : Q_i.\mathcal{E}_i$ in
$\mathcal{N}$, as well as the Environment agent $E$. We take a local state $l_i
\in L_i$ of agent $i$
to be a tuple of vector variables $(\vec{x}, \vec{y}, \vec{s}, \vec{q}, pc)$
defined as follows:
\begin{itemize}
    \item Each classical input bit in ${\mathbf i}_i$ 
    is mapped to a variable $y \in l_i$ in the domain $\{0,1\}$.
    \item A bit received from an agent via the classical receive event
      c?$x$ in the event sequence $\mathcal{E}_i$ is mapped to a
      variable $x \in l_i$ in the domain $\{0,1, \bot\}$, where
      $\bot$ denotes the undefined value before communication.
    \item A variable $s \in l_i$, called \emph{signal}, represents
      the outcome of a measurement event $M^{\alpha}_{q}$ in the event
      sequence $\mathcal{E}_i$, where $q$ is the measured qubit and
      $\alpha$ is a measurement angle. A signal can attain values
      $\{0,1, \bot\}$, where $\bot$ denotes the undefined value of the
      signal before the agent executes the measurement.
    \item A variable $q \in l_i$ in the domain $\{0,1,2\}$ represents
      the ownership relation between agent $\mathbf{A}_i$ and qubit
      $q$ with the following meaning: if $\mathbf{A}_i$ is not in
      possession of $q$, i.e., $q \notin Q_i$, then we take $q =
      0$. If $\mathbf{A}_i$ owns the qubit $q$, i.e., $q \in Q_i$,
      then $q = 1$ or $q = 2$. The former value represents that
      $\mathbf{A}_i$ does not know the exact state of the qubit, the
      latter value represents that she knows it. We assume that the
      agent knows the state of the qubit once she measures it or
      prepares it in a specific state. This is motivated as there is
      classical information involved in both cases. However, the agent
      loses this knowledge when she sends the qubit to another agent,
      as it is no longer in her possession, or entangles it with
      another qubit. Note that correction commands preserve knowledge
      because they are deterministic actions that neither entangle
      nor separate qubits.
    \item 
$pc \in l_i$ is a counter for the events in the event
      sequence $\mathcal{E}_i$ executed by agent $\mathbf{A}_i$.\\
\end{itemize}

\textbf{Example 1.} Consider the specification of QTP in DMC as given in
Section~\ref{sub:QTP}. The local state of Alice is described by the tuple
$l_A=(s_1,s_2,q_1,q_2,q_3,pc)$, and similarly the local state of Bob is
$l_B=(x_1,x_2,q_1,q_2,q_3,pc)$. In the initial state Alice owns the input qubits
$q_1$ and $q_2$ in the entangled pair, while Bob owns the qubit $q_3$, and
neither of them knows anything about their qubits. Alice has not yet measured any
qubit nor has she sent anything to Bob. The program counters of both agents point
to the first event in their event sequences. All this is captured in variable
assignments $(\bot,\bot,1,1,0,1)$ for Alice and
$(\bot,\bot,0,0,1,1)$ for Bob.\\

A local state $l_E \in L_E$ of the Environment represents the quantum state
$\sigma$ of the network.
$l_E$ is a tuple of vector variables $(\vec{q}, \vec{q}', \vec{e}, gc)$ defined
as follows:
\begin{itemize}
    \item We divide the global quantum state at any given time into
      the smallest possible substates - individual qubits and/or
      systems of entangled qubits - such that these are in pure
      states, i.e., they can be represented as a vector in a Hilbert
      space. We generate the reachable quantum state space of the
      network using the small-step rules for patterns and enumerate
      all such encountered substates. Thus, every reachable substate
      has an associated name $qs_n$, $n \in \mathbb{N}$.
    \item For every qubit $q \in \mathcal{N}$ we introduce a variable $q
      \in l_E$. The domain of $q$ is the set of names of quantum
      states that $q$ may attain in any run of the protocol, together
      with the value $\bot$ indicating that the qubit is not in a pure
      state but entangled with other qubits.
    \item Similarly, for every system of entangled qubits we introduce a
      variable $e \in l_E$. The domain of $e$ is the set of names
      of quantum states that the system may attain, together with the
      value $\bot$ indicating that either the system is not in a pure
      state or its qubits are not entangled.
    \item Each variable $q$ and $e$ is assigned a name if only if
      they 
      are pure and 
      cannot be further separated. Otherwise, they are assigned
      the value $\bot$. The global state $\sigma$ is then
      the tensor product of these substates.
    \item In addition, we make use of an auxiliary variable $q'$ for
      each qubit $q \in \mathcal{N}$ recording the name of its initial
      state, and introduce the global counter $gc \in l_E$ that
      increases with every action in the network. This is used to
      track the global time and to enumerate the configurations in
      $\mathcal{C}_\mathcal{N}$ according to their occurrence in the
      path.\\
\end{itemize}

\textbf{Example 2.}
The global quantum state of QTP is represented in the local state of the
Environment $E$ as the tuple $l_E =
(q_1,q_2,q_3,q'_1,q'_2,q'_3,e_{23},e_{123},gc)$.
The initial state of the input qubit $q_1$ is $[a, b]^T$, for $a, b \in
\mathbb{C}$. We assume that it is not equal to states $[1,0]^T$ and $[0,1]^T$ of
the standard basis, nor to states $\frac{1}{2}[\sqrt{2},\sqrt{2}]^T$ and
$\frac{1}{2}[\sqrt{2},-\sqrt{2}]^T$ of the measurement basis. In these cases
there are fewer states, but the procedure is analogous. Table~\ref{tab:enum}
shows the enumeration of substates occurring in all possible runs of the network,
as Alice and Bob execute quantum commands according to QTP.
For instance, the initial state of the network is
$(qs_1,\bot,\bot,qs_1,\bot,\bot,qs_2,\bot,1)$. Note that only the input qubit
$q_1$ and the system of two entangled qubit $e_{23}$ have assigned named states.
This is because the individual qubits $q_2$ and $q_3$ are not in a pure state and
the system of all three qubit $e_{123}$ can be further separated. Indeed, the
whole quantum state can be expressed as the tensor product $[a, b]^T \otimes
\frac{1}{2}[1,1,1,-1]^T$, or by using names $qs_1 \otimes qs_2$.

\begin{table}
{\small
    \centering
    \begin{tabular}{|c|c|c|c|}
        \hline
        \textbf{Action} & \textbf{Qubit/Entangled System} & \textbf{State} & \textbf{Name} \\
        \hline
        \multirow{2}{*}{Initially} 
            & $q_1$
            & $[a, b]^T$
            & $qs_1$ \\
            \cline{2-4}
            & $e_{23}$
            & $\frac{1}{2}[1,1,1,-1]^T$
            & $qs_2$ \\
        \hline
        $E_{12}$ & $e_{123}$ & $\frac{1}{2}[a,a,a,-a,b,b,-b,b]^T$ & $qs_3$ \\
        \hline
        \multirow{4}{*}{$M^{0}_{1}$} 
            & \multirow{2}{*}{$q_1$}
                & $\frac{1}{2}[\sqrt{2},\sqrt{2}]^T$
                & $qs_4$ \\
                & & $\frac{1}{2}[\sqrt{2},-\sqrt{2}]^T$
                & $qs_5$ \\
            \cline{2-4}
            & \multirow{2}{*}{$e_{23}$}
                & $\frac{1}{2}[a+b,a+b,a-b,-a+b]^T$
                & $qs_6$ \\
                & & $\frac{1}{2}[a-b,a-b,a+b,-a-b]^T$
                & $qs_7$ \\
        \hline
        \multirow{6}{*}{$M^{0}_{2}$} 
            & \multirow{2}{*}{$q_2$}
                & $\frac{1}{2}[\sqrt{2},\sqrt{2}]^T$
                & $qs_4$ \\
                & & $\frac{1}{2}[\sqrt{2},-\sqrt{2}]^T$
                & $qs_5$ \\
            \cline{2-4}
            & \multirow{4}{*}{$q_3$}
                & $[a, b]^T$
                & $qs_1$ \\
                & & $[a, - b]^T$
                & $qs_8$ \\
                & & $[b, a]^T$
                & $qs_9$ \\
                & & $[-b, a]^T$
                & $qs_{10}$ \\
        \hline
        $X^{x_2}_3Z^{x_1}_3$ & $q_3$ & $[a, b]^T$ & $qs_1$ \\
        \hline
    \end{tabular}
    \caption{Enumeration of quantum substates in the evolution of QTP.}
    \label{tab:enum}
}
\end{table}

\subsection{Transition Rules}
\label{sub:transitions}

Events in the event sequence $\mathcal{E}_i$ of agent $\mathbf{A}_i$ are mapped
into actions in $Act_i$.
Actions are executed according to a protocol function $P_i$ and their effects are
described by evolution functions $t_i$ and $t_E$ depending on whether the
classical state of agent $\mathbf{A}_i$ changes, or the quantum state $\sigma$ of
the system changes, or both.
Before introducing the mapping for events, note that DMC is a probabilistic
calculus, whereas IS have a Boolean semantics. We deal with this issue by
allowing all admissible transitions, abstracting away from the probability
distribution. As a result,
we lose the ability to reason about the probability of reaching a state. However,
this is not an issue for us as we need to reason about non-probabilistic
properties only as the choice of the language $\mathcal{L}$
demonstrates.

Note also that the execution of a pattern $\mathcal{P}$ in DMC occurs in a single
transition step and depends on the big-step semantics of the pattern (see Rule
\ref{eq:r1}). However, we handle transitions at the level of individual commands
of $\mathcal{P}$, and so the execution depends on the small-step semantics of
patterns and may span across several time steps.
This leads to a finely grained state space. In the rest of this section we
present the actions, the protocols, and the evolution functions associated with
the classical and quantum communication and the quantum commands presented in
Section \ref{sub:DMC}.


\textbf{Classical rendez-vous}. Assume that agent $\mathbf{A}_i$ sends the value
of $y$ to agent $\mathbf{A}_j$ who stores it in $x$, specified in DMC as
$\Gamma_i, \mathbf{A}_i:Q_i.\text{c!}y$ and $\Gamma_j,
\mathbf{A}_j:Q_j.\text{c?}x$, and that this is the $v$th (resp.~$w$th) event in
$\mathcal{E}_i$ (resp.~$\mathcal{E}_j$). We translate this by considering the
actions $snd\_j\_y_0$ and $snd\_j\_y_1$ in the set $Act_i$ of actions for agent
$i$, and action $rcv\_i\_x$ in $Act_j$. The protocol functions are:
{\small
\begin{align}
    P_i(l_i) &= \{snd\_j\_y_0\} \text{, if } pc = v \land y = 0, \nonumber \\
    P_i(l_i) &= \{snd\_j\_y_1\} \text{, if } pc = v \land y = 1, \nonumber \\
    P_j(l_j) &= \{rcv\_i\_x\} \text{, if } pc = w. \nonumber 
\end{align}}
The configuration transition, described by Rule \ref{eq:r2}, is translated into
the following evolution functions for the agents $i$ and $j$:
{\small
\begin{align}
    t_i(l_i,Act_i,Act_j) &= pc \mapsto pc + 1 \text{, if } (Act_i = snd\_j\_y_0 \lor Act_i = snd\_j\_y_1) \land Act_j = rcv\_i\_x, \nonumber \\
    t_j(l_j,Act_i,Act_j) &= pc \mapsto pc + 1 \land x \mapsto 0 \text{, if } Act_i = snd\_j\_y_0 \land Act_j = rcv\_i\_x, \nonumber \\
    t_j(l_j,Act_i,Act_j) &= pc \mapsto pc + 1 \land x \mapsto 1 \text{, if } Act_i = snd\_j\_y_1 \land Act_j = rcv\_i\_x. \nonumber
\end{align}}
The rationale behind the above equations is that when agents perform paired
send/receive actions at the same time step, their program counters are
incremented, and variable $x$ of agent $\mathbf{A}_j$ is assigned the transmitted
value.

\textbf{Quantum communication}. Assume that agent $\mathbf{A}_i$ sends a qubit $q
\in Q_i$ to agent $\mathbf{A}_j$, described as $\Gamma_i,
\mathbf{A}_i:Q_{i}.\text{qc!}q$ and $\Gamma_j, \mathbf{A}_j:Q_{j}.\text{qc?}q$,
and that this is the $v$th (resp.~$w$th) event in $\mathcal{E}_i$
(resp.~$\mathcal{E}_j$). We introduce actions $qsnd\_j\_q$ and $qrcv\_i\_q$ in
$Act_i$ and $Act_j$ respectively. The protocol functions are:
{\small
\begin{align}
    P_i(l_i) &= \{qsnd\_j\_q\} \text{, if } pc = v, \nonumber \\
    P_j(l_j) &= \{qrcv\_i\_q\} \text{, if } pc = w. \nonumber
\end{align}}
Rule~\ref{eq:r3} defines the configuration transition in terms of sets of qubits
$Q_i$ and $Q_j$. When $\mathbf{A}_i$ sends the qubit $q$, it is removed from her
set, and when $\mathbf{A}_j$ receives $q$, it is added to her set. This is
translated into IS by the evolution functions:
{\small
\begin{align}
    t_i(l_i,Act_i,Act_j) &= pc \mapsto pc + 1 \land q \mapsto 0 \text{, if } Act_i = qsnd\_j\_q \land Act_j = qrcv\_i\_q, \nonumber \\
    t_j(l_j,Act_i,Act_j) &= pc \mapsto pc + 1 \land q \mapsto 1 \text{, if } Act_i = qsnd\_j\_q \land Act_j = qrcv\_i\_q. \nonumber
\end{align}}

This means that when both agents concurrently execute the respective quantum
communication events, their local program counters are incremented, and the
ownership of the qubit changes, i.e., $\mathbf{A}_i$ is no longer in possession
of $q$ while $\mathbf{A}_j$ owns it but does not know its state.

\textbf{Corrections}. The events $X_q^s$ and $Z_q^s$ differ only in their matrix
representations, so we describe them together. Assume that agent $\mathbf{A}_i$
executes the Pauli operator $X$ or the Pauli operator $Z$ on a qubit $q$ at step
$v$ of $\mathcal{E}_i$ if signal $s = 1$, otherwise she skips the event. This
scenario has the following DMC description: $\Gamma_i,\mathbf{A}_i:q \uplus
R_i.U^s_q$, with $U^s_q \in \{X^s_q, Z^s_q\}$.  We introduce actions $x\_q$ and
$z\_q$ in $Act_i$, and since the agent applies the event conditionally, we also
include the action $skip$. In the rest of the description we refer to both
actions $x\_q$ and $z\_q$ as $u\_q$. The protocol function is then given as:
{\small
\begin{align}
    P_i(l_i) &= \{skip\} &\text{if } pc = v \land s= 0; \nonumber \\
    P_i(l_i) &= \{u\_q\} &\text{if } pc = v \land s = 1. \nonumber
\end{align}}
For example, we have the following ground protocol function for Bob in QTP:
{\small
\begin{align*}
    P_B(l_B) &= \{skip\} \text{, if } pc = 3 \land x_1 = 0; &P_B(l_B) &= \{skip\} \text{, if } pc = 4 \land x_2 = 0;  \\
    P_B(l_B) &= \{z\_q_3\} \text{, if } pc = 3 \land x_1 = 1; &P_B(l_B) &= \{x\_q_3\} \text{, if } pc = 4 \land x_2 = 1.
\end{align*}}
The small-step semantics for corrections is defined as $\sigma, \Gamma_i
\overset{U_q^s} \longrightarrow U_r^{s_{\Gamma_i}} \sigma, \Gamma_i$. Assume that
the qubit $q$ is in system $e$, which again may be just $q$ or some entangled
system. The local state of the agent $\mathbf{A}_i$ changes only through the $pc$
increment.  We define the evolution functions as:
{\small
\begin{align}
    t_E(l_E,Act_i) &= gc \mapsto gc + 1 \land e \mapsto qs_y \text{, if } e = qs_x \land Act_i= u\_q; \nonumber \\
    t_i(l_i, Act_i) &= pc \mapsto pc +1 \text{, if } Act_i = u\_q \lor Act_i = skip; \nonumber
\end{align}}
where $qs_x$ (resp.~$qs_y$) is the name of the state before (resp.~after) the
execution. The ground evolution function of $E$ in QTP with respect to Bob's
corrections $X^{x_2}_3Z^{x_1}_3$ is given as the following equations
corresponding to measurement outcomes $x_1 x_2 \mapsto 10$, $x_1 x_2 \mapsto 01$,
and $x_1 x_2 \mapsto 11$ respectively. Note that in the last case Bob executes
both actions $z\_q_3$ and $x\_q_3$ sequentially, while in the first two cases he
executes only one of them and skips the other.
{\small
\begin{align*}
    t_E(l_E,Act_B) &= gc \mapsto gc + 1 \land q_3 \mapsto qs_1 \text{, if } q_3 = qs_8 \land Act_B= z\_q_3; \\
    t_E(l_E,Act_B) &= gc \mapsto gc + 1 \land q_3 \mapsto qs_1 \text{, if } q_3 = qs_9 \land Act_B= x\_q_3; \\
    t_E(l_E,Act_B) &= gc \mapsto gc + 1 \land q_3 \mapsto qs_9 \text{, if } q_3 = qs_{10} \land Act_B= z\_q_3.
\end{align*}}

\textbf{Entanglement}. Assume that agent $\mathbf{A}_i$ applies at step $v$ of
$\mathcal{E}_i$ the entanglement operator $E_{qr}$ on qubits $q$ and $r$. The DMC
definition of the agent in this case is $\Gamma_i,\mathbf{A}_i:q,r \uplus
R_i.E_{qr}$. Since this event is independent of signals, we add only one
corresponding action $ent\_q\_r$ to $Act_i$ and define the following protocol
function:  $P_i(l_i) = \{ent\_q\_r\} \text{, if } pc = v$. The small-step rule
for entanglement is given as $\sigma, \Gamma_i \overset{E_{qr}}\longrightarrow
{}^{C}Z_{qr}\sigma,\Gamma_i$, where $^{C}Z_{qr}$ is the controlled-$Z$ operator
realising the entanglement. Since we divide the global state $\sigma$ into its
smallest pure substates, we have two possible situations. In the first case $q
\in e'$ and $r \in e''$, where $e'$ and $e''$ are isolated qubits, distinct
entangled systems, or combination of both. The resulting entangled system $e$ is
the union of the two systems $e'$ and $e''$, and we define the evolution function
of the Environment $E$ as:
{\small
\begin{equation*}
    t_E(l_E,Act_i)  =  gc \mapsto gc + 1 \land e \mapsto qs_z \land e' \mapsto \bot \land e'' \mapsto
    \bot, \text{ if } e' = qc_x \land e'' = qc_y \land Act_i = ent\_q\_r;
\end{equation*}}
where $qs_x$, $qs_y$, and $qs_z$ are the names of the quantum states in which the
systems $e'$, $e''$ are during the execution of the event, and $e$ after the
execution. For instance, the ground evolution function in QTP for Alice's
entanglement $E_{12}$ is:
{\small
\begin{eqnarray*}
    t_E(l_E,Act_A) &  = & gc \mapsto gc + 1 \land e_{123} \mapsto qs_3 \land q_1 \mapsto \bot \land e_{23} \mapsto
    \bot,\\ 
 & & \ \ \ \ \ \ \ \ \ \ \ \ \ \ \ \  \text{if } q_1 = qc_1 \land e_{23} = qc_2 \land Act_A = ent\_q_1\_q_2. \nonumber
\end{eqnarray*}}
Note that there may be many possible combinations of various states for $e'$ and
$e''$, and we have to define the evolution function for all of them. In the
second case the qubits $q$ and $r$ are part of the same system $e$ and we simply
have the evolution function:
{\small
\begin{equation}
    t_E(l_E,Act_i) = gc \mapsto gc + 1 \land e \mapsto qs_y \text{, if } e = qs_x \land Act_i = ent\_q\_r; \nonumber
\end{equation}}
where $qs_x$ (resp.~$qs_y$) is the name of the state before (resp.~after) the
execution. In both cases the local state of agent $\mathbf{A}_i$ is updated as
follows:
{\small
\begin{equation}
    t_i(l_i,Act_i) = pc \mapsto pc +1 \land q \mapsto 1 \land r \mapsto 1 \text{, if } Act_i = ent\_q\_r. \nonumber
\end{equation}}
This equation states that the counter of $\mathbf{A}_i$ is incremented and the
agent loses any knowledge about the state of $q$ and $r$ she might have had,
since neither qubit is in a pure state anymore.

\textbf{Measurement}. This is a complex event modifying the quantum state of the
network as well as the local states of agents. Suppose that agent $\mathbf{A}_i$
in step $v$ of $\mathcal{E}_i$ measures her qubit $q$ in the
$\{\mathinner{|{+_{\alpha}}\rangle}, \mathinner{|{-_{\alpha}}\rangle}\}$ basis,
specified in DMC as $\Gamma_i,\mathbf{A}_i:q \uplus
R_i.{}^{t}[M_q^{\alpha}]^{s}$, where $s$ and $t$ are signals. A measurement is a
stochastic event and may also depend on signals $s$ and $t$. We express this
non-determinism by associating two actions to a given local state $l_i$ of agent
$i$. However, due to a possible dependency on signals $s$ and $t$, there are four
different sets of actions and protocol rules. We list them in
Table~\ref{tab:measurement_action_rules}, where $\emptyset$ means that the
measurement does not depend on a particular signal.


\begin{table}
{\small
    \centering
    \begin{tabular}{|c|c|c|}
        \hline
        \multirow{2}{*}{$\emptyset$ $\emptyset$}
                & Actions & $m\_q\_+_{\alpha}$, $m\_q\_-_{\alpha}$ \\
            \cline{2-3}
            & Protocol & $P_i(l_i(g)) = \{m\_q\_+_{\alpha}, m\_q\_-_{\alpha}\}$, if $pc = v$ \\
        \hline
        \multirow{3}{*}{$s$ $\emptyset$}
                & Actions & $m\_q\_s_0\_+_{\alpha}$,
                            $m\_q\_s_0\_-_{\alpha}$,
                            $m\_q\_s_1\_+_{\alpha}$,
                            $m\_q\_s_1\_-_{\alpha}$ \\
            \cline{2-3}
            & \multirow{2}{*}{Protocol} & $P_i(l_i(g)) = \{m\_q\_s_0\_+_{\alpha},
                m\_q\_s_0\_-_{\alpha}$\}, if  $pc = v \land s = 0$\\
            & & $P_i(l_i(g)) = \{m\_q\_s_1\_+_{\alpha},
                m\_q\_s_1\_-_{\alpha}\}$, if  $pc = v \land s = 1$ \\
        \hline
        \multirow{3}{*}{$\emptyset$ $t$}
                & Actions & $m\_q\_t_0\_+_{\alpha}$,
                            $m\_q\_t_0\_-_{\alpha}$,
                            $m\_q\_t_1\_+_{\alpha}$,
                            $m\_q\_t_1\_-_{\alpha}$ \\
            \cline{2-3}
            & \multirow{2}{*}{Protocol} & $P_i(l_i(g)) = \{m\_q\_t_0\_+_{\alpha},
                m\_q\_t_0\_-_{\alpha}$\}, if  $pc = v \land t = 0$\\
            & & $P_i(l_i(g)) = \{m\_q\_t_1\_+_{\alpha},
                m\_q\_t_1\_-_{\alpha}\}$, if  $pc = v \land t = 1$ \\
        \hline
        \multirow{6}{*}{$s$ $t$}
                & \multirow{2}{*}{Actions} & $m\_q\_s_0\_t_0\_+_{\alpha}$,
                                             $m\_q\_s_0\_t_0\_-_{\alpha}$,
                                             $m\_q\_s_0\_t_1\_+_{\alpha}$,
                                             $m\_q\_s_0\_t_1\_-_{\alpha}$, \\
                                        & & $m\_q\_s_1\_t_0\_+_{\alpha}$,
                                            $m\_q\_s_1\_t_0\_-_{\alpha}$,
                                            $m\_q\_s_1\_t_1\_+_{\alpha}$,
                                            $m\_q\_s_1\_t_1\_-_{\alpha}$ \\
            \cline{2-3}
            & \multirow{4}{*}{Protocol} & $P_i(l_i(g)) = \{m\_q\_s_0\_t_0\_+_{\alpha},
                                          m\_q\_s_0\_t_0\_-_{\alpha}$\}, if  $pc = v \land s = 0 \land t = 0$\\
                                      & & $P_i(l_i(g)) = \{m\_q\_s_0\_t_1\_+_{\alpha},
                                          m\_q\_s_0\_t_1\_-_{\alpha}$\}, if  $pc = v \land s = 0 \land t = 1$\\
                                      & & $P_i(l_i(g)) = \{m\_q\_s_1\_t_0\_+_{\alpha},
                                          m\_q\_s_1\_t_0\_-_{\alpha}$\}, if  $pc = v \land s = 1 \land t = 0$\\
                                      & & $P_i(l_i(g)) = \{m\_q\_s_1\_t_1\_+_{\alpha},
                                          m\_q\_s_1\_t_1\_-_{\alpha}$\}, if  $pc = v \land s = 1 \land t = 1$\\
        \hline
    \end{tabular}
    \caption{Actions and protocol rules for various degree of dependency
        of measurements.}
    \label{tab:measurement_action_rules}
}
\end{table}

The following two transitions are defined in the small-step semantics for the
measurement event: $\sigma, \Gamma_i
\xrightarrow{^{t}[M^{\alpha}_{q}]^{r}} _{\lambda}
\mathinner{\langle{+_{\alpha_{\Gamma_i}}}|}_{q}
\sigma, \Gamma_i[0/q]$ and $\sigma,\Gamma_i
\xrightarrow{^{t}[M^{\alpha}_{r}]^{s_1}} _{\lambda}
\mathinner{\langle{-_{\alpha_{\Gamma_i}}}|}_{q}
\sigma, \Gamma_i[1/q]$. This is the source of non-determinism in the transition
system, but we do not consider the
probability $\lambda$ as long as it is non-zero.

There are again four types of evolution functions. They differ in the computation
of quantum states, and since we give only the general rules, here we describe the
evolution functions only for the independent measurement, i.e., when $s = t =
\emptyset$. As far as the translation rules are concerned, the other three types
differ only in the names of the actions and the actual names of quantum states.
We can translate them analogously.

We now consider two cases where both measurement outcome are possible. First, for
the measurement of an isolated qubit $q$ we define the evolution function of the
Environment $E$ as follows:
{\small
\begin{align}
    t_E(l_E,Act_i) &= gc \mapsto gc + 1 \land q \mapsto qs_{+_\alpha} \text{, if } q = qc_x\land Act_i = m\_q\_+_{\alpha}; \nonumber \\
    t_E(l_E,Act_i) &= gc \mapsto gc + 1 \land q \mapsto qs_{-_\alpha} \text{, if } q = qc_x \land Act_i = m\_q\_-_{\alpha}; \nonumber
\end{align}}
where $qs_{+_\alpha}$ and $qs_{-_\alpha}$ are names of the
$\{\mathinner{|{+_{\alpha}}\rangle},\mathinner{|{-_{\alpha}}\rangle}\}$
measurement basis. If the qubit $q$ is part of an entangled system $e$, then the
system becomes separated on measurement. The measured qubit $q$ collapses and the
rest of qubits form a new system $e'$. We define the evolution function as
follows:
{\small
\begin{align}
    t_E(l_E,Act_i) &= gc \mapsto gc + 1 \land q \mapsto qs_{+_\alpha} \land e \mapsto 
    \bot \land e' \mapsto qs_y \text{, if } e = qc_x \land Act_i = m\_q\_+_{\alpha}, \nonumber \\
    t_E(l_E,Act_i) &= gc \mapsto gc + 1 \land q \mapsto qs_{-_\alpha} \land e \mapsto 
    \bot \land e' \mapsto qs_z \text{, if } e = qc_x \land Act_i = m\_q\_-_{\alpha}. \nonumber
\end{align}}
In both cases the measurement outcome is assigned to a signal variable $s'$ of
agent $\mathbf{A}_i$ and her evolution function is given by:
{\small
\begin{align}
    t_i(l_i,Act_i) &= pc \mapsto pc + 1 \land s' \mapsto 0 \land q \mapsto 2 \text{, if } Act_i = m\_q\_+_{\alpha}; \nonumber \\
    t_i(l_i,Act_i) &= pc \mapsto pc + 1 \land s' \mapsto 1\land q \mapsto 2 \text{, if } Act_i = m\_q\_-_{\alpha}. \nonumber
\end{align}}
For instance, consider the first measurement that Alice performs in QTP. All
three qubits are entangled together and therefore measuring the input qubit $q_1$
causes separation of the system $e_{123}$ into two parts, $q_1$ and $e_{23}$, and
has two possible outcomes. Both have probability $\lambda = 0.5$, but we do not
take this into account since all we require is that they are non-zero, therefore
the respective transitions are admissible. We have the following ground evolution
functions for the Environment and Alice:
{\small
\begin{align*}
    t_E(l_E,Act_A) &= gc \mapsto gc + 1 \land q_1 \mapsto qs_4 \land e_{123} \mapsto 
    \bot \land e_{23} \mapsto qs_6 \text{, if } e_{123} = qc_3 \land Act_A = m\_q_1\_+_{\alpha}; \\
    t_E(l_E,Act_A) &= gc \mapsto gc + 1 \land q_1 \mapsto qs_5 \land e_{123} \mapsto 
    \bot \land e_{23} \mapsto qs_7 \text{, if } e_{123} = qc_3 \land Act_A = m\_q_1\_-_{\alpha};  \\
    t_A(l_A,Act_A) &= pc \mapsto pc + 1 \land s_1 \mapsto 0 \land q_1 \mapsto 2 \text{, if } Act_A = m\_q_1\_+_{\alpha}; \\
    t_A(l_A,Act_A) &= pc \mapsto pc + 1 \land s_1 \mapsto 1\land q_1 \mapsto 2 \text{, if } Act_A = m\_q_1\_-_{\alpha}.
\end{align*}}
In the case that the measured qubit is in a state that coincides with one of the
states of the measurement basis, there is only one possible outcome and we need
to prevent reaching an impossible state.
The translation of the transition function in case that
a measurement outcome has zero probability
requires modification of the evolution functions. We only show the case when
measuring $\mathinner{|{-_{\alpha}}\rangle}$ is impossible. The evolution
function of the Environment is given as:
{\small
\begin{align*}
    t_E(l_E,Act_i) &= gc \mapsto gc + 1 \land q \mapsto qs_{+_\alpha} \text{, if }
    q = qc_{+_\alpha} \land (Act_i = m\_q\_+_{\alpha} \lor Act_i = m\_q\_-_{\alpha}).
\end{align*}}
The Environment ``signals'' that the measurement of $q$ in a quantum state $qs_x$
has only one possible outcome. We introduce action $env_{x}$ in $Act_E$ and
define the following protocol function: $P_E(l_E) = \{env_{x}\}$, if $q  = qs_x$.
The evolution function of agent $i$ is then defined as:
{\small
\begin{align*}
    t_i(l_i,Act_i,Act_E) &= pc \mapsto pc + 1 \land s' \mapsto 0 \land q \mapsto 2, 
     \text{ if } Act_i = m\_q\_+_{\alpha} \lor (Act_i = m\_q\_-_{\alpha} \land Act_E =  env_x), \\
     t_i(l_i,Act_i,Act_E) &= pc \mapsto pc + 1 \land s' \mapsto 1 \land q \mapsto
    2 \text{, if } Act_i = m\_q\_-_{\alpha} \land Act_E \neq  env_x.
\end{align*}
}


\subsection{Correctness Proof}

We now show that the translation $f$ defined in the previous section is sound,
that is, $f$ preserves the truth conditions of formulae defined in the language
$\mathcal{L}$ introduced in Section~\ref{sec:IS} from the set of atomic
propositions $AP = \{ x = y, q_i = q_j \}$.  In \cite{ckfqcr} the truth
conditions for
the atoms in $AP$ in a configuration $C$ of a network $\mathcal{N}$ are given as
follows:
\begin{eqnarray*}
  (C, \mathcal{N}) \models x = y & \text{iff} & \text{there
    are agents } i, j \text{ such that } \Gamma_i(x) =
  \Gamma_i(y);\\ (C, \mathcal{N}) \models q_i = q_j &
  \text{iff} & \text{the global quantum state } \sigma \text{ is such that } \sigma = q_i
  = q_j.
\end{eqnarray*}

Intuitively, $x = y$ holds iff the bits denoted by $x$ and $y$ are equal.  Also,
$q_i = q_j$ holds iff the qubits denoted by $q_i$ and $q_j$ are equal.  We can
prove the following result on the translation $f$ and the language $\mathcal{L}$.
\begin{theorem} \label{equiv}
  For every formula $\phi \in \mathcal{L}$,
\begin{eqnarray*}
  (C, \mathcal{N}) \models \phi & \text{iff} &  (f(\mathcal{N}), f(C))  \models \phi
\end{eqnarray*}
\end{theorem}

\begin{proof}
  The proof is by induction on the length of $\phi$. For reasons of
  space, we only provide a sketch of the proof.  If $\phi$
  is an atomic formula, then $\phi$ is of the form $\mathfrak{a} =
  \mathfrak{b}$, where $\mathfrak{a}$ and $\mathfrak{b}$ are both
  either bits or qubits.  By the definition of $f(C)$ in
  Section~\ref{sub:States} we can easily check that $(C, \mathcal{N})
  \models \mathfrak{a} = \mathfrak{b}$ iff $(f(\mathcal{N}),f(C))
  \models \mathfrak{a} = \mathfrak{b}$. Thus, the base case holds.
The inductive case for propositional connectives $\neg$ and $\lor$ is
straightforward.

 If $\phi = E X \psi$, then by the translation of events in the event
 sequence $\mathcal{E}$ into actions in $Act$ defined in
 Section~\ref{sub:transitions}, we can see that two configurations $C,
 C' \in \mathcal{N}$ are in the temporal relation induced by
 $\mathcal{E}$, or $C \xrightarrow{\mathcal{N}} C'$, iff their
 translations $f(C), f(C') \in f(\mathcal{N})$ are in the temporal
 relation induced by $Act$, or $f(C) \xrightarrow{f(\mathcal{N})}
 f(C')$. The result then follows by the induction hypothesis.
The inductive case for the other temporal operators is similar.

If $\phi = K_i \psi$, then by the definition of the local state $l_i$ of an agent
$i$ in Section~\ref{sub:States}, we have that $l_i(f(C)) = l'_i(f(C'))$ iff
$\Gamma_i = \Gamma'_i$ and $\mathcal{E}_i = \mathcal{E}'_i$, that is, $C
\sim^{\mathcal{N}}_i C'$ iff $f(C) \sim^{f(\mathcal{N})}_i f(C')$. Also in this
case the result follows by the induction hypothesis.
This completes the sketch. 
\end{proof}

Theorem~\ref{equiv} allows us to check whether a specification $\phi \in
\mathcal{L}$ is satisfied in a network $\mathcal{N}$ by verifying $\phi$ in the
corresponding interpreted system $f(\mathcal{N})$.

\section{Implementation and Evaluation}
\label{sec:implementation}

In this section we present an implementation of the formal map above. {\sc
dmc2ispl}\footnote{ The source code is available from
\url{http://www.doc.ic.ac.uk/~pg809/dmc2ispl.tar.gz}} is a source-to-source
compiler, written in C++ and using GNU Octave libraries for matrix operations. {\sc
dmc2ispl} translates a protocol specified in a machine-readable DMC input format
into an ISPL program. The code generated is then run by MCMAS, which in turn
reports on the specification requirements of the protocol.

We modified DMC, so it can be read by the compiler. The adaptation closely
follows the syntax of the original DMC, but also reflects some features of ISPL.
A DMC file consists of five modules: a set of \emph{agents}, a set of
\emph{qubits}, whose initial state is explicitly declared, a set of \emph{groups}
of agents that are used in formulae involving group modalities, a set of
\emph{formulae} to be verified, and a set of \emph{macros} that allow agents to
perform complex quantum operations in a single time step. The declaration of an
agent consists of a set of input qubits, a set of {\em a priori} known qubits, a
set of classical inputs, and a set of events the agent executes. For
illustration, the DMC code snippet for QTP can be found in
Listing~\ref{list:teleport}.

\begin{lstlisting}[caption={QTP.dmc},
    label={list:teleport}]
-- AGENTS
Alice: {1,2},
       {},
       {},
       {c!(Bob, s2), c!(Bob, s1), Me(2,0,-,-,s2), Me(1,0,-,-,s1), En(1, 2)};
       
Bob: {3},
     {},
     {},
     {cX(3, x2), cZ(3, x1), c?(Alice, x2), c?(Alice, x1)};

-- QUBITS
1: ?;
2, 3: {(0.5, 0), (0.5, 0), (0.5, 0), (-0.5, 0)};

-- FORMULAE
AF {3 = init(1)};
!EF K (Alice, {3}) and !EF K (Bob, {3});
AF K(Bob, {3 = init(1)});
!EF K(Alice, {3 = init(1)});
\end{lstlisting}
    
{\sc dmc2ispl} has the architecture of a standard compiler. It consists of the
three following components: a module for parsing and validating the DMC input
file, a module for generating the reachable quantum state space, and a module for
generating the ISPL output file. Essentially, since MCMAS does not support matrix
arithmetic, the compiler is responsible for computation of the reachable quantum
state space, enumeration of encountered quantum states, and generation of the
evolution function of the global quantum system. Quantum states of a $n$-qubit
system are represented as $2^n \times 1$ complex matrices and unitary operators
and measurement projections as $2^n \times 2^n$ sparse complex matrices.
After the elimination of the global phase, whenever two identical state matrices
are encountered during the evolution of the quantum state of the $n$-qubit
system, they have assigned the same name.
MCMAS then works with these enumerations.

We used the compiler to verify QTP, as well as the Quantum Key Distribution
(QKD)~\cite{qcbobt}, and the Superdense Coding (SDC)~\cite{sdc} protocol against
the properties from the reference papers~\cite{ckfqcr, raqk}.
Table~\ref{tab:properties} summarises these properties. We discuss QTP in more
detail.

The figure~\ref{fig:teleport} gives a graphical representation of the possible
configurations in the QTP network. Note that configurations are parametrised by
measurement outcomes and the quantum input $\mathinner{|{\psi}\rangle}$.
The first formula in QTP section of Table~\ref{tab:properties} states that the
$\mathcal{N}_{TP}$ network is correct, since the state of Bob's qubit $q_3$ will
eventually be equal to the initial state of Alice's qubit $q_1$.
The second formula states that neither agent knows the actual quantum state of
the qubit $q_3$ at any point of the computation. The third formula states that
Bob eventually knows that the state of his qubit $q_3$ is equal to the initial
state of qubit $q_1$. The last formula states that Alice never knows this fact.

Interestingly, while~\cite{raqk} states that all four formulae are true in the
model, {\sc mcmas} evaluated the last formula to false. The reason is that even
though Alice cannot distinguish configuration
$C^{00}_3(\mathinner{|{\psi}\rangle})$ from
$C^{00}_4(\mathinner{|{\psi}\rangle})$, the atom $q_3 = init(q_1)$ holds in both
configurations as Bob does not apply any correction for measurement outcomes $s_1
s_2 \mapsto 00$, and so the quantum state of the system is invariant along this
path. This shows the importance of an automated algorithmic approach to
verification as opposed to a hand-made inspection.

\begin{figure}
    \centering
    \includegraphics{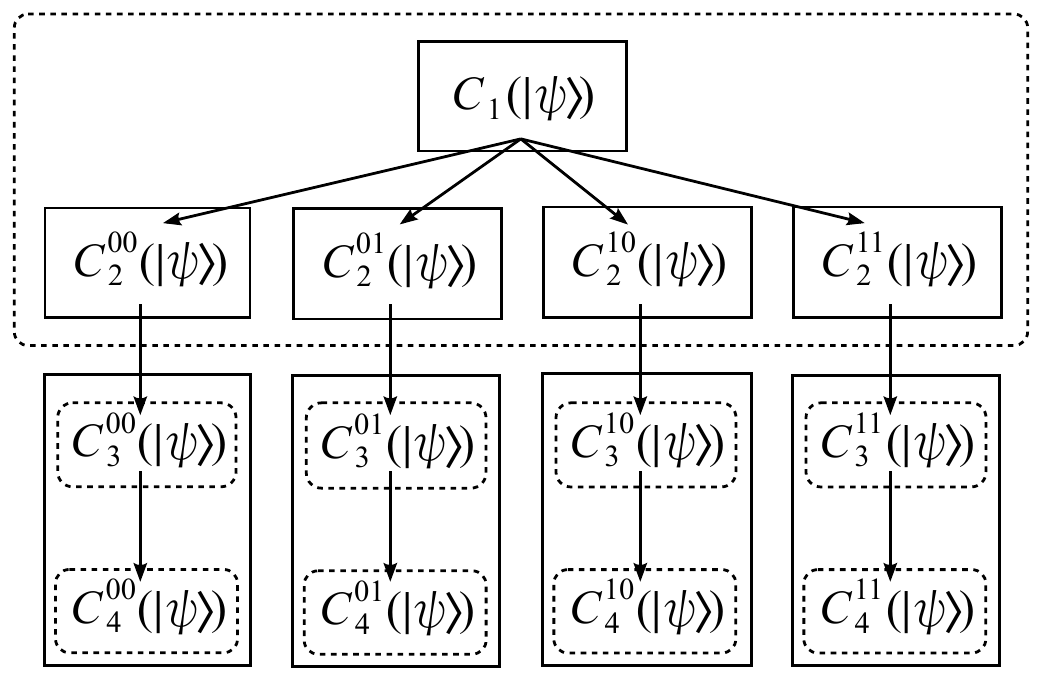}
    \caption{The epistemic accessibility relations of Alice and
        Bob in the QTP network.}
    \label{fig:teleport}
\end{figure}

\begin{table}
{\small
    \centering
    \begin{tabular}{|c|c|c|}
        \hline
        \textbf{Protocol} & \textbf{Formula} & \textbf{Reading} \\
        \hline
        \multirow{4}{*}{QTP} 
            & $AF (q_3 = init(q_1))$
            & $q_3$ eventually equals to initial $q_1$ \\
            \cline{2-3}
            & $\lnot EF K_{\mathbf A}(q_3 = \mathinner{|{\psi}\rangle}) \land \lnot EF K_{\mathbf B}(q_3 = \mathinner{|{\psi}\rangle})$
            & neither ${\mathbf A}$ nor ${\mathbf B}$ ever knows state of $q_3$ \\
            \cline{2-3}
            & $AF K_{\mathbf B}(q_3 = init(q_1))$
            & ${\mathbf B}$ eventually knows $q_1$ was teleported \\
            \cline{2-3}
            & $\lnot EF K_{\mathbf A}(q_3 = init(q_1))$
            & ${\mathbf A}$ never knows $q_1$ was teleported \\
        \hline
        \multirow{2}{*}{QKD} &  $\alpha_A = \alpha_B \rightarrow AF (K_{\mathbf{A}}(s_1=s_2)\land K_{\mathbf{B}}(s_1=s_2))$ 
                             & success if ${\mathbf A}$ \& ${\mathbf B}$ used the same basis \\
            \cline{2-3}
            & $\alpha_A \neq \alpha_B \rightarrow \lnot EF (K_{\mathbf{A}}(s_1=s_2)\lor K_{\mathbf{B}}(s_1=s_2))$ 
            & failure if ${\mathbf A}$ \& ${\mathbf B}$ used different bases \\
        \hline
        \multirow{3}{*}{SDC} &  $AF (s_1 = y_1 \land s_2 = y_2)$ & ${\mathbf B}$ eventually receives the inputs of ${\mathbf A}$ \\
            \cline{2-3}
            & $AF K_{\mathbf B}(s_1 = y_1 \land s_2 = y_2)$ & ${\mathbf B}$ eventually knows the inputs \\
            \cline{2-3}
            & $\lnot EF K_{\mathbf A}K_{\mathbf B} (s_1 = y_1 \land s_2 = y_2)$ & ${\mathbf A}$ never knows the fact above \\
        \hline
    \end{tabular}
    \caption{Verified properties of QKD and SDC protocols.}
    \label{tab:properties}
}
\end{table}

\begin{table}
{\small
    \centering
    \begin{tabular}{|c|c|c|c|c|c|c|}
        \hline
        \multirow{2}{*}{\textbf{Protocol}} &
        \multicolumn{2}{c|}{\textbf{Reachable States}} &
        \multicolumn{2}{c|}{\textbf{Memory (kB)}} &
        \multicolumn{2}{c|}{\textbf{Time (s)}} \\
        \cline{2-7}
        & \texttt{{\sc dmc2ispl}} & \texttt{{\sc mcmas}} & \texttt{{\sc dmc2ispl}} & \texttt{{\sc mcmas}} &
           \texttt{{\sc dmc2ispl}} & \texttt{{\sc mcmas}} \\
        \hline
        QTP & 40   & 108  & 7184 & 6068 & 0.015 & 0.066 \\
        QKD & 53   & 348  & 7240 & 6119 & 0.016 & 0.014 \\
        SDC & 4239 & 2192 & 8132 & 6279 & 0.112 & 0.407 \\
        \hline
    \end{tabular}
    \caption{Verification results for QTP, QKD and SDC protocols.}
    \label{tab:protocol_performance}
}
\end{table}

We conclude with some performance considerations. The tests were carried out on a
32-bit Fedora 12 Linux machine with a 2.26GHz Intel Core2 Duo processor and
2.9GiB RAM as follows: first, we translated the DMC specification into the
corresponding ISPL code using the compiler, then we analysed the resulting code
using {\sc mcmas}. Table~\ref{tab:protocol_performance} reports the results for
the three protocols. It can be seen that all protocols were verified very
quickly. This is due to their small state space and the limited number of
entangled qubits involved.

However, the amount of required resources grows exponentially for a constant
increase in the number of entangled qubits. Additionally, measuring a quantum
system using many different measurement angles results in many unique quantum
states, which in turn requires a large number of enumeration values and an
extensive evolution function. This affects the verification of a quantum protocol
by MCMAS. We analysed several experimental protocols to test the limits of the
tool. The results showed that protocols with up to $10^7$ reachable
classical states
and 20 entangled qubits can be realistically verified.

\section{Conclusion}
\label{sec:conclusion}

In this paper we presented a methodology for the automated verification of
quantum distributed systems via model checking. We defined a translation from DMC
to IS, so that {\sc mcmas} can be used to verify protocols specified in DMC.
Even though the translation does not take into account stochastic properties of
quantum protocols, in the sense that we abstract away from the underlying
probability distribution, many useful non-probabilistic properties can still be
verified as shown in reference papers~\cite{ckfqcr, raqk}.
We implemented the methodology in a source-to-source compiler and adapted the DMC
formalism to be used as an input language for the compiler. Several quantum
protocols were translated and their temporal epistemic properties were
successfully checked with {\sc mcmas}.

Given the universality of the underlying Measurement
Calculus~\cite{tmc}, the expressive power of DMC in terms of available quantum
operations is complete. However, DMC does not support any control flow statement
for the classical part of protocols.  This is one of the two major limitations of
the technique, although it can be solved by a suitable extension of the language.
Another limitation results from the state space explosion and cannot be easily
overcome since the quantum simulator requires exponential time and space on a
classical computer.

\bibliographystyle{eptcs}
\bibliography{qapl12}



\end{document}